\documentclass[journal]{IEEEtran}
\usepackage{cite}

\ifCLASSINFOpdf
\usepackage[pdftex]{graphicx}
\DeclareGraphicsExtensions{.pdf,.jpeg,.png}
\else
\usepackage[dvips]{graphicx}
\DeclareGraphicsExtensions{.eps}
\fi
\usepackage{epsfig,epsf,epstopdf,graphicx}
\usepackage[cmex10]{amsmath}
\usepackage{amssymb}
\usepackage{amsthm }
\usepackage{nicefrac}
\usepackage{hyperref}
\usepackage{xcolor}
\hypersetup{colorlinks=true}
\usepackage{tabularx}

\newtheorem{lemma}{Lemma}
\theoremstyle{theorem}

\theoremstyle{definition}

\theoremstyle{plain}

\theoremstyle{plain}

\newcommand{\nvec}{{\bf{n}}}

\newcommand{\vvec}{{\bf{v}}}

\newcommand{\xvec}{{\bf{x}}}
\newcommand{\yvec}{{\bf{y}}}
\newcommand{\zvec}{{\bf{z}}}
\newcommand{\zerovec}{{\bf{0}}}

\newcommand{\Real}{\mathbb{R}}
\newcommand{\Expv}{{\mathbb{E}}}

\newcommand{\etab}{{\mbox{\boldmath $\eta$}}}
\newcommand{\phib}{{\mbox{\boldmath $\phi$}}}

\newcommand{\Ab}{\mathbf{A}}

\newcommand{\Ib}{\mathbf{I}}

\newcommand{\Wb}{\mathbf{W}}

\newcommand{\norm}{\vert \vert}
\newcommand{\Bnorm}{\Big\| }

\newcommand{\lzero}{$\ell_{0}$}
\newcommand{\lone}{$\ell_{1}$}
\newcommand{\ltwo}{$\ell_{2}$}
\newcommand{\lp}{$\ell_{p}$}

\newcommand{\SaS}{{$S\alpha S$ }}

\newcommand{\CMN}{{\mathop{\mathrm{CMN}}}}

\usepackage{algorithm}
\usepackage{multirow}
\usepackage{algcompatible}
\usepackage[bottom]{footmisc}

\PassOptionsToPackage{dvipsnames}{xcolor}
\usepackage{tikz}
\usetikzlibrary{calc}

\ifCLASSOPTIONcompsoc
  \usepackage[caption=false,font=normalsize,labelfont=sf,textfont=sf]{subfig}
\else
  \usepackage[caption=false,font=footnotesize]{subfig}
\fi
\usepackage{fixltx2e}
\usepackage{afterpage}
\usepackage{float}
\usepackage{stfloats}
\makeatletter

\begin{document}

%
\title{Impulsive Noise Robust Sparse Recovery via Continuous Mixed Norm}

\author{Amirhossein~Javaheri,
Hadi~Zayyani,~\IEEEmembership{Member,~IEEE,}
Mario~A.~T.~Figueiredo,~\IEEEmembership{Fellow,~IEEE,}
        and~Farrokh~Marvasti,~\IEEEmembership{Life Senior Member,~IEEE}
\thanks{A.~Javaheri and F.~Marvasti are with the Department
of Electrical Engineering, Sharif University of Technology, Tehran, Iran (e-mail: javaheri\_amirhossein@ee.sharif.edu, marvasti@sharif.edu).}
\thanks{H.~Zayyani is with the Department
of Electrical and Computer Engineering, Qom University of Technology (QUT), Qom, Iran (e-mail: zayyani@qut.ac.ir).}
\thanks{M.~A.~T.~Figueiredo is with Instituto de Telecomuni\c{c}ac\~{o}es and Instituto Superior T\'{e}cnico, Lisboa, Portugal~(e-mail: mario.figueiredo@tecnico.ulisboa.pt).}
\vspace{-0.5cm}}

\maketitle
\thispagestyle{plain}
\pagestyle{plain}

\begin{abstract}
This paper investigates the problem of sparse signal recovery in the presence of additive impulsive noise. The heavy-tailed impulsive noise is well modelled with stable distributions. Since there is no explicit formulation for the probability density function of \SaS distribution, alternative approximations like Generalized Gaussian Distribution (GGD) are used which impose \lp-norm fidelity on the residual error. In this paper, we exploit a Continuous Mixed Norm (CMN) for robust sparse recovery instead of \lp-norm.  We show that in blind conditions, i.e., in case where the parameters of noise distribution are unknown, incorporating CMN can lead to near optimal recovery. We apply Alternating Direction Method of Multipliers  (ADMM) for solving the problem induced by utilizing CMN for robust sparse recovery. In this approach, CMN is replaced with a surrogate function and Majorization-Minimization technique is incorporated to solve the problem. Simulation results confirm the efficiency of the proposed method compared to some recent algorithms in the literature for impulsive noise robust sparse recovery.

\end{abstract}


%
\IEEEpeerreviewmaketitle

\section{Introduction}
\label{sec:Intro}
\IEEEPARstart{I}{n} a Compressive Sensing (CS) problem, the objective is to reconstruct a sparse signal from compressed or lower-dimensional linear measurements. This has found many applications in various fields in signal processing within the last decade \cite{survey15}. In simple words, a CS problem is solved by optimizing a cost function, which in case of unconstrained optimization approach, is generally comprised of two terms. The first term penalizes the residual error signal and \ltwo-norm is usually used for this purpose. This is an optimal choice in the presence of AWGN noise. The second term is used to promote sparsity constraint on the representation coefficients and \lzero-norm \cite{IHT_Surr} or \lone-norm \cite{ChenDS99} is often employed for this regularization. However, sometimes, there exists some heavy-tailed impulsive noise in the measurements which degrades the performance of CS reconstruction. Hence, there are a class of robust sparse recovery algorithms introduced. In \cite{Stud12}, the impulsive noise is treated as a sparse vector and a joint sparse recovery method is proposed to reconstruct the original signal. In \cite{Lore13}, an iterative hard thresholding algorithm is suggested based on Lorentzian-norm as the fidelity criterion. Another algorithm introduced in \cite{Yang11}, applies Alternating Direction Method of Multipliers (ADMM) to solve the \lone-norm minimization problem. In addition to the popular \ltwo-norm fidelity, \lp-norm is incorporated for robust sparse recovery, specifically in case of impulsive noise corruption. Robust greedy pursuit algorithms based on \lp-correlation in \lp-space are also devised in \cite{Zeng16}. Recently, an \lp-\lone~optimization approach named Lp-ADM is proposed in \cite{Wen17}, in which \lp-norm is used as penalty on the residual error  and \lone-norm is employed for sparsity regularization. There are, in addition, some robust Bayesian sparse recovery algorithms proposed in \cite{Ji13} and \cite{Shang15}.
 
In this paper, we use the Continuous Mixed Norm (CMN) proposed in \cite{Zayy14} as the error penalty function. The CMN mixes all \lp-norms within the range ($p_s\le p\le p_f$) by weighted integration over $p$.
 Since the closed form PDF, modelling the impulsive noise with \SaS distribution does not exist, alternative PDFs are employed for the approximation of this heavy-tailed distribution. In \cite{Wen17}, a zero-mean Generalized Gaussian Distribution (GGD) with shape parameter $p$ ($0<p<2$) is used to model the PDF of impulsive noise which leads to an \lp-\lone~minimization algorithm. The shape parameter is chosen to be equal to $\alpha$ which requires this parameter to be known. In  this paper, we consider the case where the distribution parameters are unknown. Treating these parameters as unobserved latent variables, an Expectation Maximization (EM) method is used for progressively optimal estimation of the measurement signal corrupted with impulsive noise. This results in an optimization problem in which CMN is used as the fidelity criterion.
To solve the problem corresponding to this  formulation, we use ADMM to spilt the problem into simpler sub-problems. We also apply Majorization-Minimization (MM) technique proposed in \cite{Figu07}. The advantage of using CMN compared to \lp-norm (as in Lp-ADM) is that it needs not to fine tune the shape parameter $p$ and thus achieves better performance in blind settings.

\section{Problem Formulation}
\label{sec:ProblemForm}

The problem explored in this paper is to recover a sparse signal given its linear random measurements corrupted with impulsive noise. Suppose $\xvec\in \Real^n$ is the original sparse signal and $\Ab \in \Real ^{m\times n}$ is the measurement matrix. The signal of linear measurements $ \Ab \xvec$, is then, added with impulsive noise signal modelled with   symmetric $\alpha$-stable i.i.d. components, i.e., each component is a 
 random variable denoted by $N$  having \SaS distribution ($\beta=0$)  with the scale parameter $\gamma$ and the location parameter $\delta=0$. We assume the PDF for this distribution is approximated with GGD as follows:
\begin{equation}
\label{eq_pdf_appr}
\!\!N \sim {\cal S} (\alpha, 0, \gamma,0), \quad f_N(n)\approx \frac{\alpha}{2\sigma_n \Gamma\left(\frac{1}{\alpha}\right)} \exp\left(- \frac{\vert n \vert ^\alpha}{\sigma_n^\alpha}\right)
\end{equation}
where  $\sigma_n$ is a constant (function of $\gamma$) and $\Gamma$ denotes the gamma function. 

The objective is to recover $\xvec$, given the noisy measurements $\yvec= \Ab\xvec+\nvec \in \Real ^{m}$.   
If the problem is blind, the parameters of the \SaS distribution modelling the impulsive noise are unknown. If we treat these unknown parameters, denoted by $\Theta$, as unobserved latent variables, we can apply EM algorithm to find the original sparse signal. The EM algorithm \cite{EM_book} is an iterative method comprising of Expectation (E) and Maximization (M) steps in each iteration, i.e. we have: 
\begin{align}  
\label{eq_EM_general}
\text{E-step:}\quad  Q(\xvec \vert \xvec^{(t)}) &=\Expv_{\Theta\vert \yvec, \xvec^{(t)}} \left[\log \left(f(\yvec, \Theta \vert \xvec) f(\xvec)\right) \right] \nonumber \\
\text{M-step:}\qquad\,\: \xvec^{(t+1)} &=\mathop{\mathrm{argmax}}_{\xvec} Q(\xvec \vert \xvec^{(t)})
\end{align}
For $\xvec^{(t+1)}$ obtained from the M-step \eqref{eq_EM_general}, it can be proven that $f(\yvec \vert \xvec^{(t+1)}) f(\xvec^{(t+1)})>f(\yvec \vert \xvec^{(t)}) f(\xvec^{(t)})$. So the algorithm gradually maximizes the posterior probability $f(\xvec \vert \yvec)$ in each iteration.
Now, we expand $\Expv_{\Theta\vert \yvec, \xvec^{(t)}} [\log f(\yvec, \Theta \vert \xvec) ]$ as follows:
\begin{align}
\label{eq_Exp_EM}
\!\!\!\Expv_{\Theta\vert \yvec, \xvec^{(t)}} [\log f(\yvec, \Theta \vert \xvec) ] &= \int_{\Theta} f(\Theta\vert \yvec, \xvec^{(t)}) \log f(\yvec, \Theta \vert \xvec)  \mathrm{d}\Theta \nonumber \\
&= \underbrace{\int_{\Theta} f(\Theta\vert \yvec, \xvec^{(t)}) \log f(\yvec\vert \Theta , \xvec)  \mathrm{d}\Theta }_{{E_1}} \nonumber \\
& +   \underbrace{ \int_{\Theta} f(\Theta\vert \yvec, \xvec^{(t)}) \log f( \Theta\vert \xvec ) \mathrm{d}\Theta }_{{E_2}}
\end{align}
where the latter equality is implied using the Bayes' rule. Now, assume the unknown noise parameters ($\Theta$) are independent from $\xvec$. Therefore, $ f( \Theta\vert \xvec )=f(\Theta)$ and the term $E_2$  may be discarded while maximizing $Q(\xvec \vert \xvec^{(t)})$ over $\xvec$.
Now, assume $\sigma_n$ and $\gamma$ are given and $\alpha$ is the only unknown parameter. Thus $\Theta = \alpha$ and the probability density $f(\yvec \vert \Theta, \xvec)$ is equal to $f(\yvec\vert \alpha , \xvec)$.
Now, using \eqref{eq_pdf_appr} for i.i.d. components of noise signal, we may write:
\begin{align*}
f(\yvec\vert \alpha , \xvec)  = \prod_{i=1}^m f_N(n_i\vert \alpha) =\left(\frac{\alpha}{2\sigma_n \Gamma\left(\frac{1}{\alpha}\right)}\right)^{\!\!m} \!\! \exp \left( -\frac{\norm \nvec \norm_\alpha^\alpha}{\sigma_n^\alpha} \right)    
\end{align*}
where $\nvec = \yvec - \Ab \xvec$. Therefore, the optimization problem in M-step \eqref{eq_EM_general} is finally reduced to:
 
\begin{align}
\label{eq_EM_reduced}
\xvec^{(t+1)} &=\mathop{\mathrm{argmax}}_{\xvec}\, Q(\xvec \vert \xvec^{(t)}) =E_1+E_2+\Expv_{\Theta\vert \yvec, \xvec^{(t)}} [\log f(\xvec)]  \nonumber\\
&\equiv \mathop{\mathrm{argmax}}_{\xvec}\, -\!\! \int_{\alpha} f(\alpha\vert \yvec, \xvec^{(t)})  \frac{\norm \yvec - \Ab \xvec\norm_\alpha^\alpha}{\sigma_n^\alpha} \mathrm{d}\alpha  +\log f(\xvec)
\end{align}


\vspace{-3ex}
\section{Proposed Algorithm}
\label{sec:proposed_Al}
It is difficult to obtain a formulation for $f(\alpha\vert \yvec, \xvec^{(t)})$, but if we ignore $\yvec$ and $\xvec^{(t)}$ while integrating over $\alpha$, we can replace this PDF with a general function of $\alpha$ denoted by $\lambda(\alpha)$. Also ignoring the iteration index $(t)$, problem \eqref{eq_EM_reduced} is transformed into a general iteration-independent optimization problem as follows:
 \begin{align}
 \label{eq_opt_general}
\hat{\xvec} = \mathop{\mathrm{argmax}}_{\xvec}\, -\!\! \int_{\alpha} \lambda(\alpha) \frac{\norm \yvec -\Ab \xvec \norm_\alpha^\alpha}{\sigma_n^\alpha}\mathrm{d}\alpha + \log f(\xvec)
 \end{align}
Now, suppose $\lambda(\alpha)$ has a support within the range $\alpha\in [p_s,p_f]$. If we also assume $\xvec$ has a Laplacian prior, it implies $\log f(\xvec)=-\mu \norm \xvec \norm_1$. Hence, problem \eqref{eq_opt_general} may be restated as:
\begin{align}
\label{eq_ML}
\hat{\xvec} = \mathop{\mathrm{argmin}}_{\xvec}\, C(\xvec)=  \ell\left(\frac{\yvec-\Ab\xvec}{\sigma_n}\right)+ \mu \norm \xvec \norm_1
\end{align}
where $\ell(\vvec) =\CMN(\vvec) = \int_{p_s}^{p_f} \lambda(p) \norm \vvec \norm_p^p \mathrm{d}p$ is what we call the continuous mixed norm \cite{Zayy14} of vector $\vvec$. In fact, we have shown that problem \eqref{eq_EM_reduced}, under  assumptions discussed in the beginning of this section, is equivalent to a sparse recovery problem in which, CMN is incorporated as fidelity criterion.
Now, inserting the auxiliary variable $\zvec = \frac{1}{\sigma_n}(\Ab\xvec - \yvec)$ into problem \eqref{eq_ML}, we can solve the following multivariate optimization problem, equivalently:
\begin{align}
\label{eq_prob_error_cmn}
\min_{\xvec ,\zvec} C(\xvec,\zvec) =\ell(\zvec)+ \mu \norm \xvec \norm_1\quad s.t. \quad \zvec = \frac{\Ab\xvec - \yvec}{\sigma_n}
\end{align}
Using the Augmented Lagrangian Method (ALM), problem \eqref{eq_prob_error_cmn} is transformed into the following unconstrained optimization problem ($\sigma>0$):
\begin{align}
\min_{\xvec ,\zvec} C^{\cal L} (\xvec ,\zvec,\etab) = \ell(\zvec)+ \mu \norm \xvec \norm_1&+  \etab^T\left( \frac{\Ab\xvec - \yvec}{\sigma_n} - \zvec \right) \nonumber \\
&+\frac{\sigma}{2} \Bnorm\frac{\Ab\xvec - \yvec}{\sigma_n} - \zvec \Bnorm ^2
\end{align}
Applying ADMM \cite{BoydADMM} to solve the this problem, we obtain an iterative solution with the following alternating steps:

\vspace{-2ex}
\subsection*{$\zvec$ update step}
The optimization problem  associating with this step is: 
\begin{align}
\label{eq_z_upd}
\zvec^{(k+1)}  &= \mathop{\mathrm{argmin}}_{\zvec} C^{\cal L} (\xvec^{(k)} ,\zvec,\etab^{(k)})  \nonumber \\
&\equiv \mathop{\mathrm{argmin}}_{\zvec} \ell(\zvec) + \frac{\sigma}{2} \Bnorm\frac{\Ab \xvec^{(k)}-\yvec}{\sigma_n}- \zvec + \frac{\etab^{(k)}}{\sigma} \Bnorm ^2
\end{align}

The solution to \eqref{eq_z_upd} is found  by applying Majorization Minimization (MM). In other words, we iteratively  optimize a set of surrogate functions instead of $\ell(\zvec)$ at each iteration. 
\begin{lemma}
\label{Lemma1}
For any $0<p\leq q$ and any real  $z_i'$, the function $s(z_i,z_i')= \vert z_i \vert ^{q} \left( \frac{p}{q}\vert z_i' \vert ^{p-q}\right)+ \left(1-\frac{p}{q}\right)\vert z_i'\vert^p$ is a surrogate function for $\vert z_i\vert ^p$ w.r.t $z_i$.
\end{lemma}
\begin{proof}
The proof is simple investigating $s(z_i,z_i')-\vert z_i \vert^p\geq 0$ using first and second order derivatives with respect to $z_i$.
\end{proof}

Assume $q \geq p_f$. Using lemma \ref{Lemma1}, we may write:
\begin{align}
\ell(\zvec) &= \sum_i \int_{p_s}^{p_f} \lambda(p) \vert z_i\vert^p  \mathrm{d}p \nonumber \\
& \leq \sum_i \int_{p_s}^{p_f}\! \lambda(p) \! \left( \vert z_i \vert ^{q} \left( \frac{p}{q}\vert z_i' \vert ^{p-q}\right)+ \left(1-\frac{p}{q}\right)\vert z_i'\vert^p \right) \! \mathrm{d}p \nonumber \\
&= \sum_i  \vert z_i \vert ^{q} \phi_{p_s,p_f}^{(q)}(z_i') + \sum_i \psi_{p_s,p_f}^{(q)}(z_i') = \ell_S(\zvec, \zvec')
\end{align}
The term $\ell_S(\zvec, \zvec')$ is a surrogate function for $\ell(\zvec)$ (note that $\ell_S(\zvec, \zvec)=\ell(\zvec)$ and $\ell_S(\zvec, \zvec') > \ell(\zvec)$ for any $\zvec\neq \zvec'$). If we assume uniform distribution for $\lambda(p)$, we obtain:
\begin{equation}
\label{eq_phi}
\phi_{p_s,p_f}^{(q)}(z_i') = \dfrac{\vert z_i' \vert ^{p_f}\left(p_f \log\vert z_i'\vert -1\right) -\vert z_i' \vert ^{p_s}\left(p_s \log\vert z_i' \vert -1\right)}{(p_f-p_s)q \vert z_i' \vert ^q \log^2\vert z_i' \vert}
\end{equation}

Now let $\zvec' = \zvec^{(k)}$; using MM technique, it suffices to solve the following optimization problem at each iteration:

\begin{align}
\label{eq_z_upd_un}
\zvec^{(k+1)}  &= \mathop{\mathrm{argmin}}_{\zvec} C^{\cal L}_S (\xvec^{(k)} ,\zvec,\etab^{k})   \\
&= \mathop{\mathrm{argmin}}_{\zvec} \ell_S(\zvec,\zvec^{(k)}) + \frac{\sigma}{2} \Bnorm\frac{\Ab \xvec^{(k)}-\yvec}{\sigma_n}- \zvec + \frac{\etab^{(k)}}{\sigma} \Bnorm ^2 \nonumber
\end{align}

 Hence, depending on the value of $q$, we will deal with the following problems:
 
I.) {$q=1$ and $p_s<p_f\leq1$}:
In this case we have:
\begin{align}
\ell_S(\zvec,\zvec^{(k)}) = \sum_i \vert z_i \vert \phi_{p_s,p_f}^{(1)} (z_i^{(k)}) + \sum_i \psi_{p_s,p_f}^{(1)} (z_i^{(k)}) 
\end{align}
Now, substituting $\phi_{p_s,p_f}^{(1)}(z_i^{(k)})$ from \eqref{eq_phi} and discarding the constant term $\sum_i  \psi_{p_s,p_f}^{(1)} (z_i^{(k)}) $, problem \eqref{eq_z_upd_un} is transformed into an \lone~minimization problem with respect to $\zvec$ where the solution is obtained via soft-thresholding operator ${\cal S}_T$ \cite{Daub04}:
\begin{align}
\label{eq_z_upd_l1}
\!\!\!\!\!\!\!\zvec^{(k+1)} ={\cal S}_{{T^{(k)}}}\left( \frac{\Ab \xvec^{(k)}-\yvec }{\sigma_n}+ \frac{\etab^{(k)}}{\sigma}\right)
\end{align}
and $T^{(k)} \!=\! \frac{ 1}{\sigma}{\phib}_{p_s,p_f}^{(1)}\!(\zvec^{(k)}\!) \!=\! \frac{ 1}{\sigma}[\phi_{p_s,p_f}^{(1)}\!(\!z_1^{(k)}\!),  \cdots ,\phi_{p_s,p_f}^{(1)}\!(\!z_m^{(k)}\!) ]^T$.

II.) {$q=2$ and $p_s<p_f\leq2$}:
The choice of $q=2$ results in quadratic formulation for $\ell_S(\zvec,\zvec^{(k)}) $. Thus, substituting $\phi_{p_s,p_f}^{(2)}(z_i^{(k)})$ from \eqref{eq_phi}, we can restate problem \eqref{eq_z_upd_un} as follows:
\begin{align*}
\zvec^{(k+1)} =\mathop{\mathrm{argmin}}_{\zvec} \zvec^T \Wb(\zvec^{(k)}) \zvec + \frac{\sigma}{2} \Bnorm\frac{\Ab \xvec^{(k)}-\yvec}{\sigma_n} - \zvec + \frac{\etab^{(k)}}{\sigma} \Bnorm ^2
\end{align*}
where $\Wb(\zvec^{(k)})=\mathrm{diag}(\phib_{p_s,p_f}^{(2)}(\zvec^{(k)}))$. This problem has a closed form solution obtained by:

\begin{align}
\label{eq_z_upd_l2}
\!\!\zvec^{(k+1)}  &= \left( \Ib+\frac{2}{\sigma} \Wb(\zvec^{(k)}) \right)^{-1} \left(  \frac{\Ab \xvec^{(k)}-\yvec}{\sigma_n} + \frac{\etab^{(k)}}{\sigma}\right)
\end{align}
Since the matrix $ \Ib+2/\sigma \Wb(\zvec^{(k)})$ is diagonal, this inverse solution is easily obtained with element-wise scalar divisions.

\vspace{-1em}
\subsection*{$\xvec$ update step}
This step has an optimization problem formulated as:
\begin{align}
\xvec^{(k+1)}  &= \mathop{\mathrm{argmin}}_{\xvec} C^{\cal L}_S (\xvec ,\zvec^{(k+1)},\etab^{(k)})  \\
& =  \mathop{\mathrm{argmin}}_{\xvec}  \frac{\sigma}{2} \Bnorm\frac{\Ab \xvec-\yvec}{\sigma_n}  - \zvec ^{(k+1)}+ \frac{\etab^{(k)}}{\sigma} \Bnorm ^2+\mu \norm \xvec \norm_1 \nonumber
\end{align}
solving this LASSO problem using MM with $\lambda_0 > \norm \Ab \norm^2/{\sigma_n^2}$ yields \cite{Daub04}:
\begin{align}
\label{eq_x_upd}
\xvec^{(k+1)}\!=  {\cal S}_{\!\frac{\mu}{\sigma \lambda_0}}\!\! \left(\!\xvec^{(k)}\!-\frac{1}{\lambda_0 \sigma_n}\Ab^{\!T} \! \left(\!\frac{\Ab \xvec^{(k)}-\yvec}{\sigma_n} -\zvec^{(k+1)} \!+\!\frac{\etab^{(k)}}{\sigma} \!\right)\!\! \right)
\end{align}

\vspace{-2em}
\subsection*{Multiplier update step}
Finally we have the update formula for $\etab$ which is:
\begin{equation}
\label{eq_eta_upd}
\etab^{(k+1)} = \etab^{(k)} + \sigma \left(\frac{\Ab \xvec^{(k+1)}-\yvec}{\sigma_n} - \zvec^{(k+1)}\right)
\end{equation}

\vspace{-1em}
\subsection*{Remark}
The weight function $\phi_{p_s,p_f}^{(q)}(z_i')$ given in equation \eqref{eq_phi} is undefined when $\vert z_i'\vert =0$. To deal with this problem, a common alternative (\cite{Chartrand08}) is to  consider a small regularizing parameter $\epsilon$ in the definition of CMN, i.e., problem \eqref{eq_prob_error_cmn} is restated as:
\begin{align*}
\min_{\xvec , \zvec} \int_{p_s}^{p_f} \lambda(p) \sum_i ( \vert z_i \vert +\epsilon)^p \mathrm{d}p+\mu \norm \xvec \norm_1\quad s.t. \quad \zvec =\frac{\Ab\xvec - \yvec}{\sigma_n}
\end{align*}  
The derivation for the $\epsilon$-regularized problem with approximation, is similar to what is obtained in $\zvec$-update step in section \ref{sec:proposed_Al}, except that we let $\vert z_i^{(k)} \vert \leftarrow\vert z_i^{(k)} \vert +\epsilon $ in the computation of $\phi_{p_s,p_f}^{(q)}(z_i^{(k)})$. 
In addition, to fasten the rate of convergence of the algorithm, we apply continuation method on the regularizing parameter $\mu$ (as proposed in \cite{Figu09}) with a minimum threshold denoted by $\mu_{\min}$.
Hence, modified steps of the  proposed algorithm named CMN-ALM, are given in Alg. \ref{Algorithm_1}. The stopping criterion, decides when to stop the algorithm. For this purpose, we choose a maximum tolerance for primal or dual residual error.

\begin{algorithm}[!h]
\caption{Proposed robust CS algorithm (CNM-ALM)}
\textbf{Set}   $0\leq p_s<p_f\leq q\leq 2,\quad \mu, \sigma, \lambda_0,\zeta>0, \quad  \mu_{\min}, \epsilon \ll 1$. \newline
\textbf{Initialize} $\etab^{(0)}=\zerovec$, $\xvec^{(0)}=\zerovec$, $\zvec^{(0)}=-\yvec$, $k=0$.
\label{Algorithm_1}
\begin{algorithmic}

\REPEAT
\STATE Update $\zvec^{(k+1)}$ using \eqref{eq_z_upd_l1} or \eqref{eq_z_upd_l2} (with $\vert z_i^{(k)} \vert \leftarrow\vert z_i^{(k)} \vert +\epsilon $)
\STATE Update $\xvec^{(k+1)}$ using \eqref{eq_x_upd}
\STATE Update $\etab^{(k+1)}$ using  \eqref{eq_eta_upd} 
\STATE Update $\mu \leftarrow \max\{\zeta \mu, \mu_{\min}\}$ and  $k \leftarrow k+1$
\UNTIL {A stopping criterion is reached}
\end{algorithmic}
\end{algorithm}
\vspace{-2ex}
%
%
%

\section{Simulation Results}
\label{sec:Simulation}
In this section, we conduct experiments to compare the reconstruction quality performance of our proposed method with some state of the art algorithms for robust sparse recovery. In particular, we use Lp-ADM, YALL1 \cite{YALL1}, BP-SEP \cite{BPSEP} and Huber-FISTA \cite{HFISTA}\footnote{The codes for these algorithms are available at \href{https://github.com/FWen/Lp-Robust-CS.git}{https://github.com/FWen/Lp-Robust-CS.git}}.  Based on different choices for $p_s$, $p_f$ and $q$, we obtain different versions of our proposed algorithm. In the following experiments we run 3 versions with $(p_s=0, p_f=1, q=1)$,  $(p_s=0, p_f=1, q=2)$ and $(p_s=0, p_f=2, q=2)$. We choose $\sigma_n=1$ and the stopping criterion is set to primal or dual error tolerance of $1e-5$ with 100 maximum iterations. The experiments and results are categorized into the following sub-sections:
\subsubsection{Average preference ratio}
\label{subsec:Succ_Rate}
 In this scenario, a set of $N_t$ $k$-sparse vectors of size $n\times 1$ are randomly generated where $N_t=60$, $n=128$ and $k=7$. For each vector $\xvec_i$, the elements of the support set are chosen uniformly at random and the values are generated according to normal distribution. This model is indeed used as approximation of the Laplacian prior for the original signal. Each sample vector is then compressed via a random Gaussian measurement matrix $\Ab_i$ of dimensions $m\times n$. We choose $m=50$ (nearly 0.4 sampling rate) and the measurement vector $\yvec_i  = \Ab_i \xvec_i $ is then added with \SaS noise with distribution parameters $\alpha$ and $\gamma$. These values are chosen from $\alpha\in \{ 0.5, 1, 1.5\}$ and $\gamma\in \{ 1e-4, 1e-3, 1e-2, 1e-1\}$.  The noisy observed vector is finally given to robust sparse recovery algorithms to obtain an estimate $\hat{\xvec}_i$.  We eventually calculate SNR performance and average the results over random turns (random $\xvec_i$ and $\Ab_i$). For this part, we only compare our results with Lp-ADM algorithm. We choose $\sigma=1$, $\mu_{\min}=5e-1$, $\zeta=0.95$, $\epsilon=1e-2$ and $\lambda_0 = 2$.  We have used the source code for Lp-ADM with default parameters. Since the problem is blind, we do not know the best choice of $p$ (\lp~norm) applied in Lp-ADM algorithm. To compare the performance of the algorithms, we plot the average SNR of Lp-ADM versus $p$. We then calculate the \emph{preference ratio}, defined as the percentage of the region $p \in (0,2)$ where the SNR curves corresponding to different versions of our proposed algorithm (these curves are constant lines versus $p$) lie above the Lp-ADM curve. These regions for each curve are highlighted in the preference diagram depicted in Fig. \ref{fig_ver_p}. In fact, this experiment demonstrates the probability that our proposed algorithms outperform Lp-ADM when the problem is blind and $p$ is chosen uniformly at random.
 Furthermore, in a blind problem, neither the original sparse signal nor the parameter $\gamma$ is known. Thus, the regularization path for specifying $\mu$ as proposed in \cite{Wen17} is not applicable.  In this case, we choose $\mu = \xi \norm \Ab^T \xvec \norm_\infty$ (as  in \cite{Figu09}) with $\xi =0.1$ for all algorithms.
Table \ref{Table_1} shows these preference ratios for different values of \SaS noise parameters. According to this table, the version of the proposed algorithm with $(p_s=0, p_f=1, q=1)$ has best performance in blind settings. Even the worst case corresponding to the version $(p_s=0, p_f=2, q=2)$ which fails at $\alpha=0.5$, outperforms Lp-ADM most of the times especially for $\alpha>0.5$ and $\gamma>1e-4$.

\begin{table}[!h]
\vspace*{-1ex}
\caption{Preference ratio (in \%) of different versions of the proposed algorithm, as defined in section \ref{subsec:Succ_Rate}, compared to Lp-ADM. The values are reported in 3-tuples corresponding to $(p_s=0, p_f=1, q=1)$,  $(p_s=0, p_f=1, q=2)$ and $(p_s=0, p_f=2, q=2)$ versions, respectively. }
\centering
\hspace*{-.3cm} \begin{tabular}{l||p{1.5cm}|p{1.5cm}|p{1.5cm}|p{1.5cm}} \hline
& $\gamma = 1e-4$  & $\gamma = 1e-3$  & $\gamma = 1e-2$  &$\gamma = 1e-1$  \\ \hline \hline 
  $\alpha=0.5$ &100\, 100\, 0 &100\, 100\, 0\, &100\, 100\, 0 &100\, 0\hspace{2ex} 0 \\ \hline 
 $\alpha=1$ &87\hspace{2ex} 0\hspace{2ex}  0 &63\hspace{2ex} 87\hspace{1ex} 100 &100\, 100\, 100 &100\, 100 0 \\ \hline 
 $\alpha=1.5$ &100\, 100\, 100 &100\, 100\, 100  &57\hspace{2ex} 78\hspace{1ex} 100 &100\, 100\, 100\\ \hline 
\end{tabular}
\label{Table_1}
\end{table}

\begin{figure}[!t]
\vspace{-7ex}
\centering
\includegraphics[width=1.9in, trim={1.8in 3.2in 2in 3in},clip]{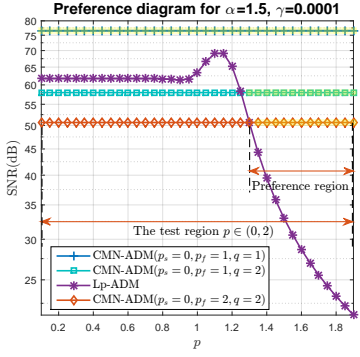}
\caption{The preference diagram of SNR versus $p$ (GGD approximation shape parameter) depicting the preference region for robust CS reconstruction (experiment \ref{subsec:Succ_Rate}). The parameters of $S\alpha S$ noise are $\alpha = 1.5,\, \gamma=1e-4$. \vspace{-2ex}}
\label{fig_ver_p}
\end{figure}

\subsubsection{The influence of noise power}
\label{subsec:Noise_level}
In this part, we would like to examine to effect of noise power on the performance of robust sparse recovery algorithms. The settings and  the parameters for the algorithms are just similar to the previous scenario except that $\alpha$ is fixed in this experiment and the SNR performance is depicted versus the parameter $\gamma$, which in some sense specifies the additive noise power. For Lp-ADM we choose $p\in \{0.5, 1, 1.5\}$ and the results are compared with those of our proposed methods as well as YALL1, BP-SEP and Huber-FISTA. Fig. \ref{fig_noise_1} shows the SNR performance of robust CS algorithms in \SaS noise versus the scale parameter $\gamma$. In Fig. \ref{fig_noise_1} we have chosen $\alpha=0.5$ and Fig. \ref{fig_noise_2} depicts the results for $\alpha=1$. As shown in these figures, the proposed algorithms clearly outperform competing algorithms specifically for $\alpha=0.5$.

\begin{figure}[!t]
\vspace{-8.1ex}
\centering
\hspace*{-1em} \subfloat[\SaS noise with $\alpha = 0.5$]{\includegraphics[width=2.1in]{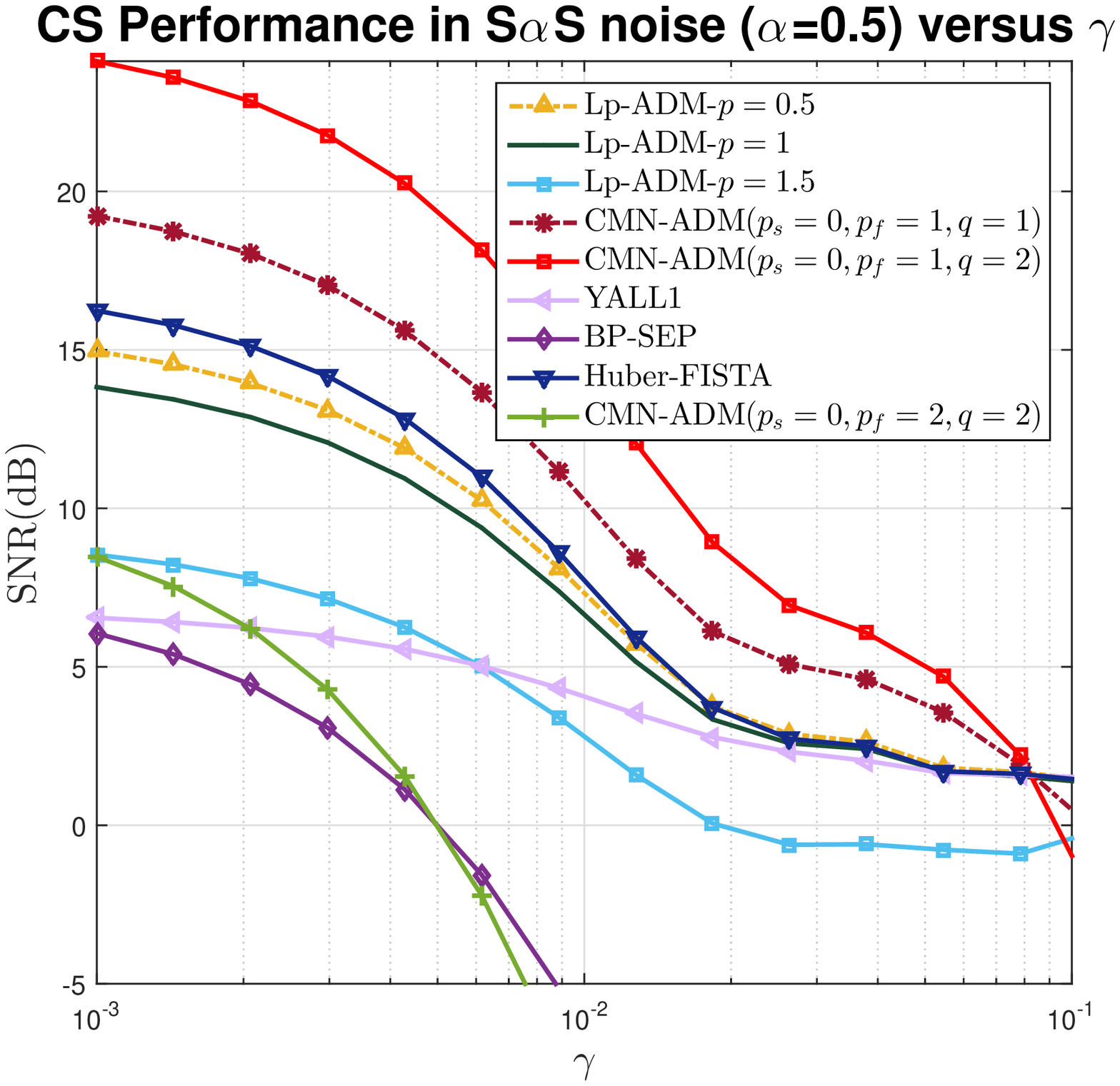}\label{fig_noise_1} } \hspace*{-1.4em} \subfloat[\SaS noise with  $\alpha = 1$]{\includegraphics[width=2.1in]{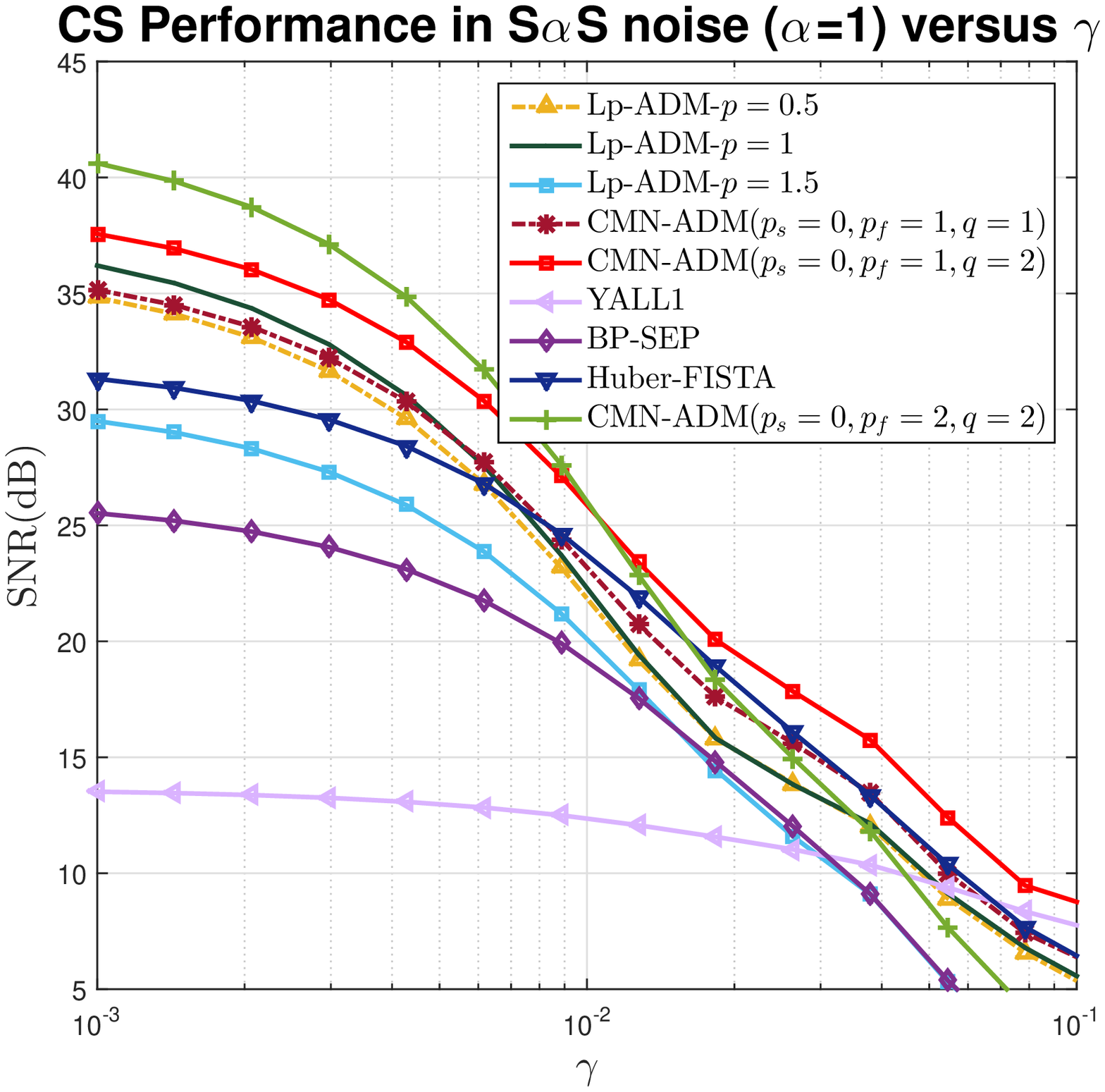}\label{fig_noise_2}  }
\caption{SNR performance of robust sparse recovery algorithms versus $\gamma$ (the scale parameter of $S \alpha S$ impulsive noise) (experiment \ref{subsec:Noise_level}).} 
\end{figure}

\begin{figure}[!t]
\vspace{-3ex}
\centering
\hspace*{-1em}  \subfloat[\SaS noise,  $\alpha = 1$, $\gamma=1e-3$]{\includegraphics[width=2.1in]{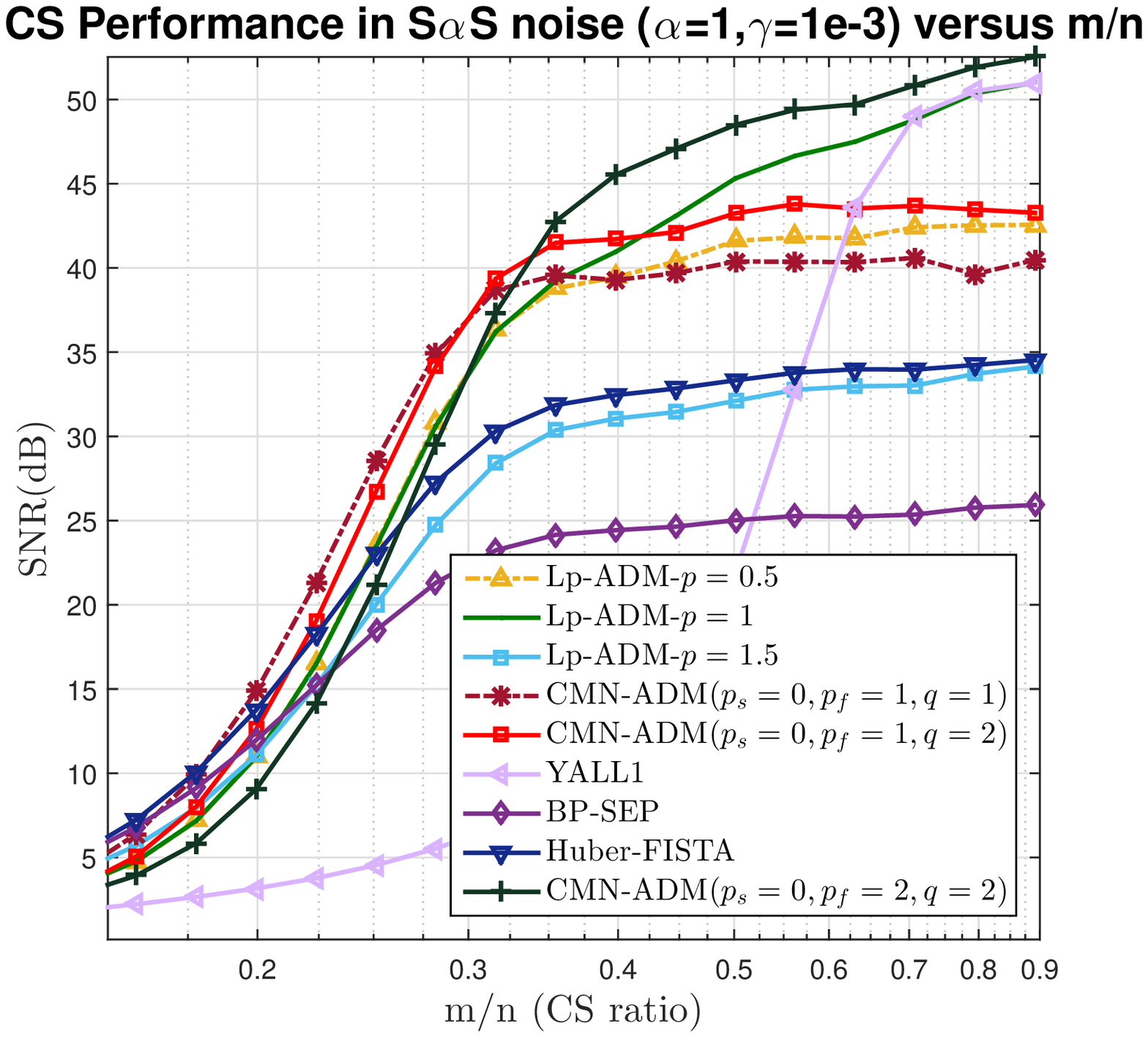}\label{fig_CS_1} }  \hspace*{-1.4em}
\subfloat[\SaS noise,  $\alpha = 1.5$, $\gamma=1e-3$]{\includegraphics[width=2.1in]{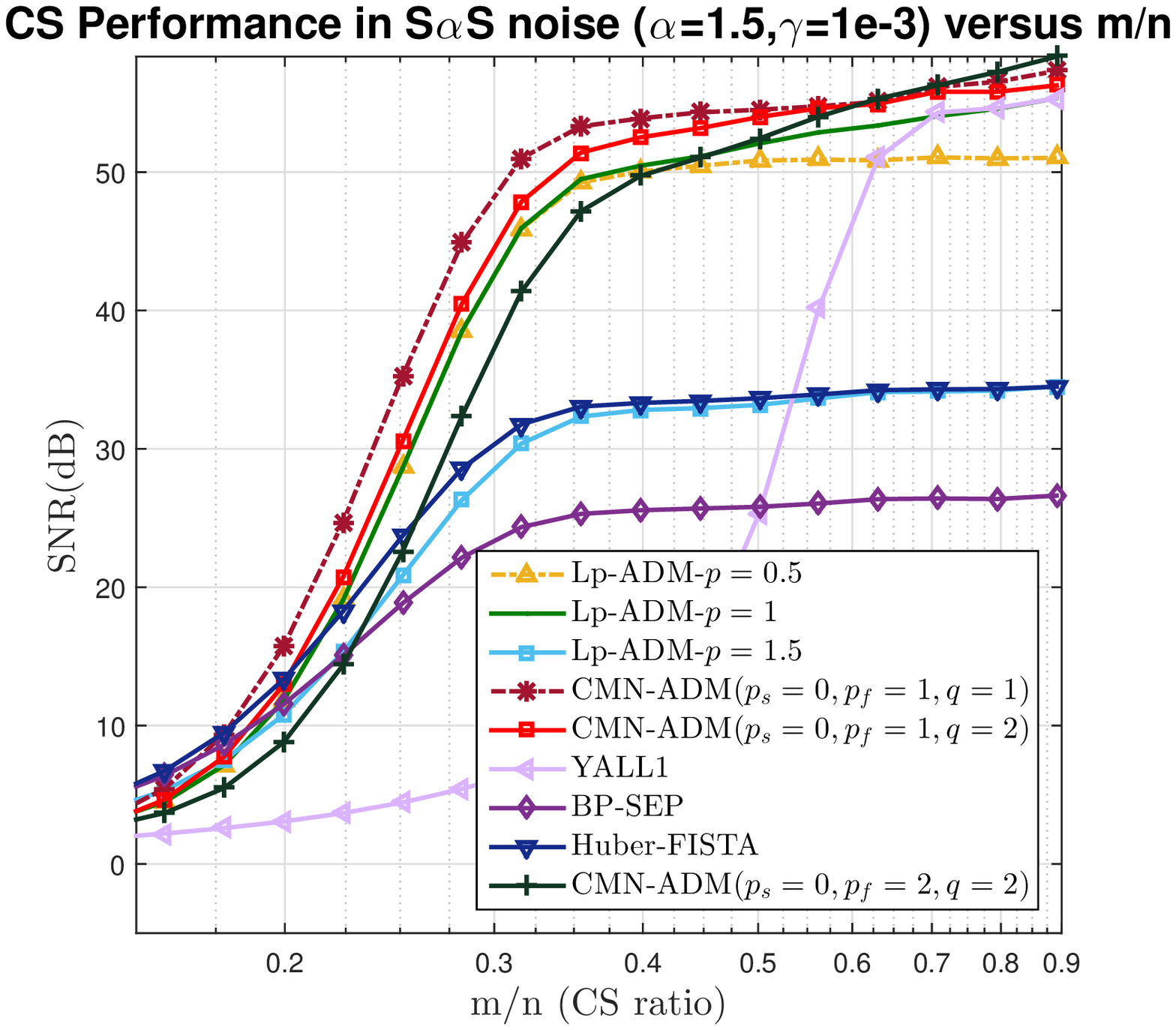}\label{fig_CS_2}  }
\caption{SNR performance of robust sparse recovery algorithms versus CS factor $m/n$ (experiment \ref{subsec:CS}). \vspace{-2ex}} 
\end{figure}

\subsubsection{The effect of CS factor}
\label{subsec:CS}
This part demonstrates the performance of the proposed algorithms in terms of the CS factor, i.e., the ratio $m/ n$ where $m$ equals the number of measurements and $n$ denotes the size of the sparse vector. Similar to previous sections (using same parameters), we generate random $8$-sparse vectors of size $n=128$, but $m$ varies from $0.1n$ to $0.9n$. The sparse signals are then corrupted with \SaS impulsive noise. We consider two cases where in the first, we let $\alpha=1$ and $\gamma = 1e-3$ and in the second, $\alpha=1.5$ and $\gamma = 1e-3$ are chosen. The SNR performance of the  sparse recovery algorithms are finally depicted versus the ratio $m/n$. Fig. \ref{fig_CS_1} and Fig. \ref{fig_CS_2} show the results for the first and the second scenario, respectively. For $\alpha=1$ the proposed algorithm with $p_s=0,p_f=2,q=2$ has the best performance while for $\alpha=1.5$ the version with $p_s=0,p_f=1,q=1$ yields more robust recovery.

%
%
\vspace*{-1ex}
\section{Conclusion}
\label{sec:conclusion}
In this paper, we explored the problem of blind sparse signal recovery in the presence of impulsive noise. We modelled the noise signal with symmetric $\alpha$-stable i.i.d.  components and we incorporated GGD to approximate the PDF of the distribution. In blind conditions, the parameters $\alpha$ and $\gamma$ of the \SaS distribution model are unknown. Treating theses parameters as unobserved latent variables, we applied EM algorithm to obtain an iterative approach to near optimal  recovery of the original signal where in each iteration, an optimization problem should be solved. Under assumptions, this problem was shown to be equivalent to a CS problem in which, a continuous mixed norm is employed as fidelity criterion. We incorporated ADMM with Majorization Minimization technique to iteratively solve the proposed CS problem. The performance of the proposed algorithm in blind CS recovery in impulsive noise, was finally examined via simulation experiments.


\begin{thebibliography}{1}

\bibitem{survey15}
Z. Zhang, Y. Xu, J. Yang, X. Li, and D. Zhang,
\newblock ``A survey of sparse representation: algorithms and applications,''
\newblock {\em IEEE Access}, vol. 3, pp. 490--530, 2015.

\bibitem{IHT_Surr}
T. Blumensath, and  M. E. Davies,
\newblock ``Iterative thresholding for sparse approximations,"
\newblock {\em Journal of Fourier Analysis and Applications}, vol. 14, no. 5, pp. 629--654, 2008.


\bibitem{ChenDS99}
S.~S. Chen, D.~L. Donoho, and M.~A. Saunders,
\newblock ``Atomic decomposition by basis pursuit,''
\newblock {\em SIAM Journal on Scientific Computing}, vol. 20, no. 1, pp.
  33--61, 1999.



\bibitem{Stud12}
C.~Studer, P.~Kuppinger, G.~Pope, and H.~Bolcskey,
\newblock ``Recovery of sparsely corrupted signals,''
\newblock {\em IEEE Trans. Inf. Theory}, vol. 58, no. 5, pp. 3115--3130, May 2012.


\bibitem{Lore13}
R.~E.~Carrillo, and K.~E.~Barner,
\newblock ``Lorentzian iterative hard thresholding: Robust compressed sensing with prior information,''
\newblock {\em IEEE Trans. on Signal Proc.}, vol. 61, no. 19, pp. 4822--4833, Oct 2013.

\bibitem{Yang11}
J.~F.~Yang, and Y.~Zhang,
\newblock ``Alternating direction algorithms for \lone-problems in compressive sensing,''
\newblock {\em SIAM J. Sci. Comput.}, vol. 33, pp. 250--278, Oct 2011.
%

\bibitem{Zeng16}
W.~J.~Zeng, H.~C.~So, and X.~Jiang,
\newblock ``Outlier-robust greedy pursuit algorithms in \lp-space for sparse approximation,''
\newblock {\em IEEE Trans. on Signal Proc.}, vol. 64, no. 1, pp. 60--75, Jan 2016.

\bibitem{Wen17}
F.~Wen, P.~Liu, Y.~Liu, R.~C.~Qiu, and W.~Yu,
\newblock ``Robust sparse recovery in impulsive noise via \lp-\lone optimization,''
\newblock {\em IEEE Trans. on Signal Proc.}, vol. 65, no. 1, pp. 105--118, January 2017.


\bibitem{Ji13}
Y. Ji, Z. Yang, and W. Li,
\newblock ``Bayesian sparse reconstruction method of compressed sensing in the presence of impulsive noise,''
\newblock {\em Circuit, Systms, and Signal Processing (CSSP)}, vol. 32, no. 6, pp. 2971--2998, 2013.

\bibitem{Shang15}
J. Shang, Z. Wang, and Q. Huang,
\newblock ``A robust algorithm for joint sparse recovery in presence of impulsive noise,''
\newblock {\em IEEE Signal Processing Letters}, vol. 22, no. 8, pp. 1166--1170, 2015.

\bibitem{Zayy14}
H. Zayyani,
\newblock ``Continuous mixed p-norm adaptive algorithm for system identification,''
\newblock {\em IEEE Signal Processing Letters}, vol. 21, no. 9, pp. 1108--1110, Sep 2014.

\bibitem{Figu07}
M.~A.~T.~Figueiredo, J.~M.~Bioucas-Diaz, and R.~D.~Nowak,
\newblock ``Majorization-Minimization algorithms for wavelet-based image restoration,''
\newblock {\em IEEE Trans. on Image Proc.}, vol. 16, no. 12, pp. 2980--2991, Dec 2007.

\bibitem{EM_book}
R. J. A. Little, D. B. Rubin
\newblock ``Statistical Analysis with Missing Data,"
\newblock {\em  Wiley Series in Probability and Statistics}, New York: John Wiley and Sons, 1987.

\bibitem{BoydADMM}
S. Boyd, N. Parikh, E. Chu, B. Peleato, and  J. Eckstein,
\newblock ``Distributed optimization and statistical learning via the alternating direction method of multipliers,"
\newblock {\em Foundations and Trends in Machine Learning}, vol. 3, no. 1, pp.1--122, 2011.

\bibitem{Daub04}
I. Duabechies, M. Defrise, and C. De Mol,
\newblock ``An iterative thresholding algorithm for linear inverse problems with a sparsity constraint,''
\newblock {\em Commun. Pure Appl. Math.}, vol. 57, no. 11, pp. 1413--1457, 2004.


\bibitem{Chartrand08}
R. Chartrand, W. Yinl,
\newblock ``Iteratively reweighted algorithms for compressive sensing,''
\newblock {\em IEEE International Conference on Acoustics, Speech and Signal Processing (ICASSP 2008)}, Las Vegas, USA, 2008.


\bibitem{Figu09}
S. J.~Wright, R. D. Nowak, and M. A. T.~Figureido,
\newblock ``Sparse reconstruction by separable approximation,''
\newblock {\em IEEE Trans. on Signal Proc.}, vol. 57, no. 7, pp. 2479--2493, 2009.

\bibitem{YALL1}
J. F. Yang and Y. Zhang, 
\newblock ``Alternating direction algorithms for \lone-problems in compressive sensing,"
\newblock {\em SIAM J. Sci. Comput.}, vol. 33, pp. 250–278, 2011. 

\bibitem{BPSEP}
C. Studer and R. G. Baraniuk, 
\newblock `` Stable restoration and separation of approximately sparse signals,"
\newblock {\em Appl. Comput. Harmon. Anal.}, vol. 37, no. 1, pp. 12–35, 2014. 

\bibitem{HFISTA}
D. S. Pham and S. Venkatesh, 
\newblock ``Efficient algorithms for robust recovery of images from compressed data,"
\newblock {\em IEEE Trans. Image Process.}, vol. 22, no. 12, pp. 4724–4737, Dec. 2013. 

\end{thebibliography}
\end{document}